\documentclass{amsart}%
\usepackage{enumerate}
\usepackage[leqno]{amsmath}
\usepackage{amssymb,amsthm}
\usepackage{hyperref}
\usepackage{graphicx}
\usepackage{amsfonts}
\usepackage{amssymb}%
\theoremstyle{plain}
\newtheorem{Theorem}{Theorem}
\newtheorem{Corollary}{Corollary}

\newtheorem{Lemma}{Lemma}
\newtheorem{Proposition}{Proposition}
\theoremstyle{Definition}
\newtheorem{Definition}{Definition}
\theoremstyle{remark}

\newtheorem{Remark}{Remark}
\numberwithin{equation}{section}
\email{jjrodriguezv@unal.edu.co}
\email{wzuniga@math.cinvestav.mx}
\subjclass[2000]{primary ; secondary .}
\begin{document}
\title[ ]{ Elliptic Pseudo-Differential Equations and Sobolev Spaces over $p$-adic Fields}
\author{J. J. Rodr\'{\i}guez-Vega}
\address{Departamento de Matem\'{a}ticas, Universidad Nacional de Colombia, Ciudad
Universitaria, Bogot\'{a} D.C., Colombia.}
\author{W. A. Z\'{u}\~{n}iga-Galindo}
\address{Centro de Investigaci\'{o}n y de Estudios Avanzados del I.P.N., Departamento
de Matem\'{a}ticas, Av. Instituto Polit\'{e}cnico Nacional 2508, Col. San
Pedro Zacatenco, M\'{e}xico D.F., C.P. 07360, M\'{e}xico. }

\begin{abstract}
We study the solutions of equations of type $f(D,\alpha)u=v$, where
$f(D,\alpha)$ is a $p$-adic pseudo-differential operator. If $v$ is a
Bruhat-Schwartz function, then there exists a distribution $E_{\alpha}$, a
fundamental solution, such that $u=E_{\alpha}\ast v$ is a solution. However,
it is unknown to which function space $E_{\alpha}\ast v$ belongs. In this
paper, we show that if $f(D,\alpha)$ is an elliptic operator, then
$u=E_{\alpha}\ast v$ belongs to a certain Sobolev space. Furthermore, we give
conditions for the continuity and uniqueness of $u$. By modifying the Sobolev
norm, we can establish that $f(D,\alpha)$\ gives an isomorphism between
certain Sobolev spaces.

\end{abstract}
\maketitle

\section{Introduction}

In recent years $p-$adic analysis has received a lot of attention due to its
applications in mathematical physics, see e.g. \cite{A-K2}, \cite{A-B-K-O},
\cite{A-B-O}, \cite{KH1}, \cite{KH2}, \cite{KO2}, \cite{R-T}, \cite{Va},
\cite{V-V-Z} and references therein. As a consequence new mathematical
problems have emerged, among them, the study of $p$-adic pseudo-differential
equations, see e.g. \cite{A-K-S}, \cite{Chuong-Co}, \cite{KO1}, \cite{KH0},
\cite{KO2}, \cite{KO3}, \cite{KO4}, \cite{KO5}, \cite{Ro-Zu}, \cite{V-V-Z},
\cite{Z-G1}, \cite{Z-G2}, \cite{Z-G3} and references therein. In this paper,
we study the solutions of $p$-adic elliptic pseudo-differential equations on
Sobolev spaces.

A pseudo-differential operator $f(D,\beta)$ is an operator of the form
\[
\left(  f(D,\alpha)\varphi\right)  (x)=\mathcal{F}_{\xi\rightarrow x}%
^{-1}\left(  |f(\xi)|_{p}^{\alpha}\mathcal{F}_{x\rightarrow\xi}\phi(x)\right)
,\text{ \ }\phi\in S,
\]
where $\mathcal{F}$ denotes the Fourier transform, $\alpha$ is a positive real
number, $S$ denotes the $\mathbb{C}$-vector space of Bruhat-Schwartz functions
over $\mathbb{Q}_{p}^{n}$, and $f(\xi)\in\mathbb{Q}_{p}[\xi_{1},\dotsc,\xi
_{n}]$. If $f(\xi)$ is a homogeneous polynomial of degree $d$ satisfying
\[
f(\xi)=0\text{ if and only if }\xi=0,
\]
then the corresponding operator is called an elliptic pseudo-differential
operator. At any case, the operator $f(D,\beta)$ is continuous and has a
self-adjoint extension with dense domain in $L^{2}(\mathbb{Q}_{p}^{n})$. This
operator is considered to be a $p$-adic analogue of a linear partial elliptic
differential operator with constant coefficients. A $p$-adic
pseudo-differential equation is an equation of type
\[
f(D,\alpha)u=v.
\]
If $v\in S$, then there exists a distribution $E_{\alpha}$, a fundamental
solution, such that $u=E_{\alpha}\ast v$ is a solution. The existence of a
fundamental solution for general pseudo-differential operators was established
by the second author in \cite{Z-G1} by adapting the proof given by Atiyah for
the Archimedean case \cite{Atiyah}. However, it is unknown to which function
space $E_{\alpha}\ast v$ belongs. In this paper, we show that if $f(D,\alpha)$
is an elliptic operator, then $u=E_{\alpha}\ast v$ belongs to a certain
Sobolev space (see Theorem \ref{mainresult}). Furthermore, we give conditions
for the continuity and uniqueness of $u$. By modifying the Sobolev norm, we
can establish that $f(D,\alpha)$ gives an isomorphism between certain Sobolev
spaces, (see Propositions \ref{prop1}, \ref{prop2} and Theorem \ref{mainII}).
Our approach is based on the explicit calculation of fundamental solutions of
pseudo-differential operators on certain function spaces and the fact that
elliptic pseudo-differential operators behave like the Taibleson operator when
acting on certain function spaces (see Theorems \ref{fund solu Tai},
\ref{fund solu ellip}).

\textbf{Acknowledgement.} The authors wish to thank the referee for his/her
careful reading of the original manuscript.

\section{Preliminary Results}

We summarize some basic facts about $p$-adic analysis that will be used in
this paper. For a complete exposition, we refer the reader to \cite{TA},
\cite{V-V-Z}.

Let $\mathbb{Q}_{p}$ be the field of the $p$-adic numbers, and let
$\mathbb{Z}_{p}$ be the ring of $p$-adic integers. For $x\in\mathbb{Q}_{p}$,
let $v(x)\in\mathbb{Z}\cup\left\{  \infty\right\}  $ denote the valuation of
$x$ normalized by the condition $v(p)=1$. By definition $v(x)=\infty$ if and
only if $x=0$. Let $|x|_{p}=p^{-v(x)}$ be the normalized absolute value. Here,
by definition $|x|_{p}=0$ if and only if $x=0$. We extend the $p$-adic
absolute value to $\mathbb{Q}_{p}^{n}$ as follows:
\[
||x||_{p}:=\text{max}\{|x_{1}|_{p},\dotsc,|x_{n}|_{p}\},\text{ for }%
x=(x_{1},\dotsc,x_{n})\in\mathbb{Q}_{p}^{n}.
\]

We define the \textit{exponent of local constancy} of $\varphi(x)\in
S(\mathbb{Q}_{p}^{n})$ as the smallest integer, $l\geq0$, with the property
that, for any $x\in\mathbb{Q}_{p}^{n}$,
\[
\varphi(x+x^{\prime})=\varphi(x)\text{ if }||x^{\prime}||_{p}\leq p^{-l}.
\]

For $x$, $y$ in $\mathbb{Q}_{p}^{n}$, we put $x\cdot y=\sum_{i=1}^{n}
x_{i}y_{i}$.

Let $\Psi$ denote an additive character of $\mathbb{Q}_{p}$, trivial on
$\mathbb{Z}_{p}$, but not on $p^{-1}\mathbb{Z}_{p}$. For $\varphi\in
S(\mathbb{Q}_{p}^{n})$, we define its Fourier transform as
\[
(\mathcal{F}\varphi)(\xi) = \int_{\mathbb{Q}_{p}^{n}} \Psi(-x\cdot\xi
)\varphi(x)\,dx,
\]
where $dx$ denotes the Haar measure of $\mathbb{Q}_{p}^{n}$ normalized in such
a way that $\mathbb{Z}_{p}^{n}$ has measure one.

We denote by $\chi_{r}$, $r\in\mathbb{Z}$, the characteristic function of the
polydisc $B_{r}(0):=(p^{r}\mathbb{Z}_{p})^{n}$. For any $\varphi\in S$, we
set
\[
r_{\varphi}:=\text{min}\{r\in\mathbb{N}\mid\varphi|_{B_{r}\left(  {0}\right)
}=\varphi(0)\}.
\]

\begin{Definition}
We set $\mathcal{L}:=\mathcal{L}(\mathbb{Q}_{p}^{n})=\{ \varphi\in S \mid
\int_{\mathbb{Q}_{p}^{n}} \varphi(x)\, dx =0\}$, and $\mathcal{W}%
:=\mathcal{W}(\mathbb{Q}_{p}^{n})$ to be the $\mathbb{C}$-vector space
generated by the functions $\chi_{r}$, $r \in\mathbb{Z}$.
\end{Definition}

We note that any $\varphi\in S$ can be written uniquely as $\varphi
_{\mathcal{L}}+\varphi_{\mathcal{W}}$, where $\varphi_{\mathcal{W}%
}=p^{r_{\varphi}n}\left(  \int_{\mathbb{Q}_{p}^{n}}\phi(x)\,dx\right)
\chi_{r_{\varphi}}\in\mathcal{W}$, and $\varphi_{\mathcal{L}}=\varphi
-\varphi_{\mathcal{W}}\in\mathcal{L}$. However, $S$ is not the direct sum of
$\mathcal{L}$ and $\mathcal{W}$. The space $\mathcal{W}$ was introduced in
\cite{Z-G2}, and $\{\mathcal{F}(\varphi)\mid\varphi\in\mathcal{L}\}$ is a
Lizorkin space of second class \cite{A-K-S}.

\subsection{Elliptic Pseudo-differential Operators}

Let $f(\xi)\in\mathbb{Q}_{p}[\xi_{1},\dotsc,\xi_{n}]$ be a nonconstant
polynomial. A pseudo-differential operator $f(D,\alpha)$, $\alpha>0$, with
symbol $|f(\xi)|_{p}^{\alpha}$, is an operator of the form
\[
\left(  f(D,\alpha)\varphi\right)  =\mathcal{F}^{-1}\left(  |f|_{p}^{\alpha
}\mathcal{F}\varphi\right)  ,
\]
where $\varphi\in S$.

\begin{Definition}
Let $f(\xi)\in\mathbb{Q}_{p}[\xi_{1},\dotsc,\xi_{n}]$ be a nonconstant
polynomial. We say that $f(\xi)$ is an elliptic polynomial of degree $d$, if
it satisfies: $(i)$ $f(\xi)$ is a homogeneous polynomial of degree $d$, and
$(ii)$ $f(\xi)=0\Leftrightarrow\xi=0$.
\end{Definition}

\begin{Lemma}
\label{Lemma 1}\cite[Lemma 1]{Z-G3} Let $f(\xi)\in\mathbb{Q}_{p}[\xi
_{1},\dotsc,\xi_{n}]$ be an elliptic polynomial of degree $d$. There exist
positive constants, $C_{0}(f)$ and $C_{1}(f)$, such that
\[
C_{0}(f)||\xi||_{p}^{d}\leq|f(\xi)|_{p}\leq C_{1}(f)||\xi||_{p}^{d},\text{ for
every }\xi\in\mathbb{Q}_{p}^{n}.
\]

\end{Lemma}

We note that if $f(\xi)$ is elliptic, then $cf(\xi)$ is elliptic for any
$c\in\mathbb{Q}_{p}^{\times}$. For this reason, we will assume from now on
that the elliptic polynomials have coefficients in $\mathbb{Z}_{p}$.

\begin{Lemma}
\label{Lemma 2}\cite[Lemma 3]{Z-G3} Let $f(\xi)\in\mathbb{Q}_{p}[\xi
_{1},\dotsc,\xi_{n}]$ be an elliptic polynomial of degree $d$. Let
$A\subset\mathbb{Q}_{p}^{n}$ be a compact subset such that $0\notin A$. Then
there exists a positive integer $m=m(A,f)$ such that $|f(\xi)|_{p}\geq p^{-m}%
$, for any $\xi\in A$. Furthermore, for any covering of $A$ of the form
$\cup_{i=1}^{L}B_{i}$, with $B_{i}=z_{i}+(p^{m}\mathbb{Z}_{p})^{n}$, we have
$|f(\xi)|_{p}=|f(z_{i})|_{p}$ for any $\xi\in B_{i}$.
\end{Lemma}

\begin{Definition}
Let $f(\xi)\in\mathbb{Z}_{p}[\xi_{1},\dotsc,\xi_{n}]$ be an elliptic
polynomial of degree $d$. We will say that $|f|_{p}^{\beta}$ is an elliptic
symbol, and that $f(D,\beta)$ is an elliptic pseudo-differential operator of
order $d$.
\end{Definition}

\subsection{Igusa's local zeta functions}

Let $g(x)\in\mathbb{Q}_{p}[x]$, $x=(x_{1},\dotsc,x_{n})$, be a non-constant
polynomial. Igusa's local zeta function associated to $g(x)$ is the
distribution
\[
\langle|g|_{p}^{s},\varphi\rangle=\int\limits_{\mathbb{Q}_{p}^{n}%
\smallsetminus g^{-1}(0)}|g(x)|_{p}^{s}\varphi(x)\,dx,
\]
for $s\in\mathbb{C}$, $\text{Re}(s)>0$, where $\varphi\in S$, and $dx$ denotes
the normalized Haar measure of $\mathbb{Q}_{p}^{n}$. The local zeta functions
were introduced by Weil and their basic properties for general $g(x)$ were
first studied by Igusa. A central result in the theory of local zeta functions
established that $|g|_{p}^{s}$ admits a meromorphic continuation to the
complex plane such that $\langle|g|_{p}^{s},\varphi\rangle$ is rational
function of $p^{-s}$ for each $\varphi\in S$. Furthermore, there exists a
finite set $\cup_{E\in\mathcal{E}}\{(N_{E},n_{E})\}$ of pairs of positive
integers such that
\[
\prod\limits_{E\in\mathcal{E}}(1-p^{-n_{E}-N_{E}s})|g|_{p}^{s}%
\]
is a holomorphic distribution on $S$. In particular, the real parts of the
poles of $|g|_{p}^{s}$ are negative rational numbers see \cite[Chap. 8]%
{Igusa}. The existence of a meromorphic continuation for the distribution
$|g|_{p}^{s}$ implies the existence of a fundamental solution for the
pseudo-differential operator with symbol $|g|_{p}^{\alpha}$, \cite{Z-G1}.

For a fixed $\varphi\in S$, we denote the integral $\langle|g|_{p}^{s}%
,\varphi\rangle$ by $Z_{\varphi}(s,g)$. In particular, $Z(s,g)=Z_{\chi_{0}%
}(s,g)$.

\begin{Lemma}
Let $f(x)\in\mathbb{Z}_{p}[x]$, $x=(x_{1},\dotsc,x_{n})$, be an elliptic
polynomial of degree $d$. Then
\[
Z(s,f)=\dfrac{L(p^{-s})}{1-p^{-ds-n}},
\]
where $L(p^{-s})$ is a polynomial in $p^{-s}$ with rational coefficients.
Furthermore, $s=-n/d$ is a pole of $Z(s,f)$.
\end{Lemma}

\begin{proof}
Let $A=\{x\in\mathbb{Z}_{p}^{n}\mid\text{ord}(x_{i})\geq d,\quad
i=1,\dotsc,n\}$, and $A^{\prime}=\{x\in\mathbb{Z}_{p}^{n}\mid\text{ord}%
(x_{i})<d,\text{ for some }i\}$. Then $\mathbb{Z}_{p}^{n}$ is the disjoint
union of $A$ and $A^{\prime}$ and
\begin{align*}
Z(s,f)  &  =\int_{A}|f(x)|_{p}^{s}\,dx+\int_{A^{\prime}}|f(x)|_{p}^{s}\,dx\\
&  =p^{-ds-n}Z(s,f)+\int_{A^{\prime}}|f(x)|_{p}^{s}\,dx,
\end{align*}
i.e., $Z(s,f)=\frac{1}{1-p^{-ds-n}}\int_{A^{\prime}}|f(x)|_{p}^{s}\,dx$. Since
$A^{\prime}$ is compact, by applying Lemma \ref{Lemma 2}, we find a covering
of $A^{\prime}=\cup_{i=1}^{L}B_{i}$, where $|f|_{p}$ is constant on each
$B_{i}$. Hence,
\[
\int_{A^{\prime}}|f(x)|_{p}^{s}\,dx=p^{-nm}\sum_{i=1}^{L}|f(z_{i})|_{p}^{s},
\]
and
\[
Z(s,f)=\dfrac{p^{-nm}\sum_{i=1}^{L}|f(z_{i})|_{p}^{s}}{1-p^{-ds-n}}.
\]

\end{proof}

\subsection{The Riesz Kernel}

We collect some well-know results about the Riesz kernel that will be used in
the next sections, we refer the reader to \cite{TA} or \cite{V-V-Z} for
further details.

The $p$-adic \emph{Gamma function} $\Gamma_{p}^{(n)}(s)$ is defined as
follows:
\[
\Gamma_{p}^{(n)}(s)=\frac{1-p^{s-n}}{1-p^{-s}}\text{, }s\in\mathbb{C}%
,{\ }s\neq0.
\]
The Gamma function is meromorphic with simple zeros at $n+\frac{2\pi i}{\ln
p}\mathbb{Z}$ and unique simple pole at $s=0$. In addition, it satisfies
\[
\Gamma_{p}^{(n)}(s)\Gamma_{p}^{(n)}(n-s)=1,\text{ for }s\notin\{0\}\cup
\{n+\frac{2\pi i}{\ln p}\mathbb{Z}\}.
\]

The \emph{Riesz kernel} $\mathit{R}_{s}$ is the distribution determined by the
function
\[
\mathit{R}_{s }(x)=\frac{||x||_{p}^{s -n}}{\Gamma_{p}^{(n)}(s) }%
,\quad\text{Re}(s)>0,{\ }s \notin n+\frac{2\pi i}{\ln p}\mathbb{Z} ,\quad
x\in\mathbb{Q}_{p}^{n}.
\]

The Riesz kernel has, as a distribution, a meromorphic continuation to
$\mathbb{C}$ given by
\begin{align*}
\left\langle \mathit{R}_{s}(x),\varphi(x)\right\rangle  &  =\frac{1-p^{-n}%
}{1-p^{s-n}}\varphi(0)+\frac{1-p^{-s}}{1-p^{s-n}}\int_{||x||_{p}>1}%
||x||_{p}^{s-n}\varphi(x)\,dx\\
&  +\frac{1-p^{-s}}{1-p^{s-n}}\int_{||x||_{p}\leq1}||x||_{p}^{s-n}%
(\varphi(x)-\varphi(0))\,dx,
\end{align*}
with poles at $n+\frac{2\pi i}{\ln p}\mathbb{Z}$. In particular, for
$\text{Re}(s)>0$,
\[
\langle\mathit{R}_{s}(x),\varphi(x)\rangle=\frac{1-p^{-s}}{1-p^{s-n}}%
\int_{\mathbb{Q}_{p}^{n}}\varphi(x)||x||_{p}^{s-n}\,dx,\quad s\notin
n+\frac{2\pi i}{\ln p}\mathbb{Z},
\]%
\begin{equation}
\langle\mathit{R}_{-s}(x),\varphi(x)\rangle=\frac{1-p^{s}}{1-p^{-s-n}}%
\int_{\mathbb{Q}_{p}^{n}}(\varphi(x)-\varphi(0))||x||_{p}^{-s-n}\,dx.
\label{1}%
\end{equation}

In the case $s=0$, by passing to the limit, we obtain
\[
\langle\mathit{R}_{0}(x),\varphi(x)\rangle:=\lim_{s\rightarrow0}\left\langle
\mathit{R}_{s}(x),\varphi(x)\right\rangle =\varphi(0),
\]
i.e., \ $\mathit{R}_{0}(x)=\delta\left(  x\right)  $, the Dirac delta
function. Therefore, $\mathit{R}_{s}\in S^{\prime}(\mathbb{Q}_{p}^{n})$, for
$s\in\mathbb{C\setminus}\left\{  n+\frac{2\pi i}{\ln p}\mathbb{Z}\right\}  $.

\begin{Remark}
\label{continuation of |x|} The distribution $||x||_{p}^{s}$, $\text{Re}%
(s)>0$, admits the following meromorphic continuation,
\begin{align*}
\left\langle ||x||_{p}^{s},\varphi(x)\right\rangle  &  =\frac{1-p^{-n}
}{1-p^{-s -n}}\varphi(0)+ \int_{||x||_{p}>1}||x||_{p}^{s} \varphi(x)\,dx\\
&  +\int_{||x||_{p}\leq1}||x||_{p}^{s} (\varphi(x)-\varphi(0))\,dx,
\quad\varphi\in S.
\end{align*}
In particular, all the poles of $||x||_{p}^{s}$ have real part equal to $-n$.
\end{Remark}

\begin{Lemma}
[{\cite[Chap. III, Theorem 4.5]{TA}}]\label{Fourier trans of Riesz} As element
of $S^{\prime}(\mathbb{Q}_{p}^{n})$, $\left(  \mathcal{F}\mathit{\ R}%
_{s}\right)  (x)$ equals $||x||_{p}^{-s}$, for $s\notin n+\frac{2\pi i}{\ln
p}\mathbb{Z}$.
\end{Lemma}

The following explicit formula will be used in the next sections.

\begin{Lemma}
\label{elliptic = Taibleson} Let $f(x)\in\mathbb{Q}_{p}[x]$, $x=(x_{1}%
,\dotsc,x_{n})$, be an elliptic polynomial of degree $d$. Then
\[
|f|_{p}^{s}=\frac{(1-p^{ds})L(p^{-s})}{(1-p^{-n})(1-p^{-ds-n})}\mathit{R}%
_{ds+n},\quad s\in\mathbb{C}%
\]
as distributions on $\mathcal{W}$. Here $L(p^{-s})$ is the numerator of
$Z(s,f)$ which is a polynomial in $p^{-s}$ with rational coefficients.
\end{Lemma}

\begin{proof}
Let $\varphi\in\mathcal{W}$, then
\[
\varphi(x)=\sum_{i}c_{i}\chi_{r_{i}}(x),
\]
where $c_{i}\in\mathbb{C}$, $r_{i}\in\mathbb{Z}$ (recall that $\mathcal{F}%
(\chi_{r})=p^{-nr}\chi_{-r}$). The action of $|f|_{p}^{s}$ on $\mathcal{F}%
\varphi$ can be explicitly described as follows:
\[
\langle|f|_{p}^{s},\mathcal{F}\varphi\rangle=\sum_{i}c_{i}\langle|f|_{p}%
^{s},p^{-nr_{i}}\chi_{-r_{i}}\rangle,
\]
but
\[
\langle|f|_{p}^{s},p^{-nr_{i}}\chi_{-r_{i}}\rangle=p^{-nr_{i}}\int
_{\mathbb{Q}_{p}^{n}}|f(x)|_{p}^{s}\chi_{-r_{i}}(x)\,dx=p^{dr_{i}s}Z(s,f),
\]
for $\text{Re}(s)>0$, thus
\[
\langle|f|_{p}^{s},\mathcal{F}\varphi\rangle=Z(s,f)\sum_{i}c_{i}p^{dr_{i}%
s},\quad\text{Re}(s)>0.
\]

On the other hand,
\[
\langle\dfrac{1-p^{ds}}{1-p^{-n}}\mathit{R}_{ds+n}, p^{-nr_{i}}\chi_{-r_{i}}
\rangle=\langle\dfrac{1-p^{-ds-n}}{1-p^{-n}}||x||_{p}^{ds}, p^{-nr_{i}}%
\chi_{-r_{i}} \rangle= p^{dr_{i} s},
\]
for every $r_{i} \in\mathbb{Z}$ and $\text{Re}(s)>0$. Then we have
\[
\langle|f|_{p}^{s},\mathcal{F}\varphi\rangle= \dfrac{1-p^{ds}}{1-p^{-n}}
Z(s,f) \langle\mathit{R}_{ds+n},\mathcal{F}\varphi\rangle,
\]
for $\text{Re}(s)>0$. Now $Z(s,f)$ and $\mathit{R}_{ds+n}$ have a meromorphic
continuation to the complex plane, therefore this formula extends to
$\mathbb{C}$. Finally, since the Fourier transform establishes a $\mathbb{C}%
$-isomorphism on $\mathcal{W}$, it is possible remove the Fourier transform symbol.
\end{proof}

\subsection{The Taibleson Operator}

\begin{Definition}
The Taibleson pseudo-differential operator $D_{T}^{\alpha}$, $\alpha>0$, is
defined as
\[
(D_{T}^{\alpha}\varphi)(x)=\mathcal{F}_{\xi\rightarrow x}^{-1}\left(
||\xi||_{p}^{\alpha}\mathcal{F}_{x\rightarrow\xi}\varphi\right)  \text{, for
}\varphi\in S.
\]

\end{Definition}

As a consequence of the Lemma \ref{Fourier trans of Riesz} and (\ref{1}), one
gets
\begin{align*}
\left(  D_{T}^{\alpha}\varphi\right)  \left(  x\right)   &  =\left(
\mathit{k}_{-\alpha}\ast\varphi\right)  \left(  x\right)  =\\
&  \frac{1-p^{\alpha}}{1-p^{-\alpha-n}}\int_{\mathbb{Q}_{p}^{n}}%
||y||_{p}^{-\alpha-n}(\varphi(x-y)-\varphi(x))\,dy.
\end{align*}

The right-hand side of \ previous formula makes sense for a wider class of
functions than $S(\mathbb{Q}_{p})$, for example, for the class $\mathfrak{E}%
_{\alpha}(\mathbb{Q}_{p}^{n})$ of locally constant functions $\varphi(x)$
satisfying
\[
\int_{||x||_{p}\geq1}||x||_{p}^{-\alpha-n}|\varphi(x)|\,dx<\infty.
\]

\begin{Remark}
As a consequence of the previous observations we may assume that the constant
functions are contained in the domain of $D_{T}^{\alpha}$, and that
$D_{T}^{\alpha}\varphi=0$, for any constant function.
\end{Remark}

\section{Fundamental Solutions for the Taibleson Operator}

We now consider the following pseudo-differential equation:
\begin{equation}
\label{eq pseudo}D_{T}^{\alpha}u=v, \quad\text{with $v\in S$}, \text{ and
}\alpha>0.
\end{equation}

We say that $E_{\alpha}\in S^{\prime}$ is a \textit{fundamental solution} of
(\ref{eq pseudo}) if $E_{\alpha}\ast v$ is a solution.

\begin{Lemma}
\label{funda + k} If $E_{\alpha}$ is a fundamental solution of
(\ref{eq pseudo}), then for any constant $c$, $E_{\alpha}+c$ is also a
fundamental solution.
\end{Lemma}

\begin{proof}
Let $E_{\alpha}$ a fundamental solution for (\ref{eq pseudo}), then
\begin{align*}
D_{T}^{\alpha}((E_{\alpha}+c)\ast v)  &  =D_{T}^{\alpha}((E_{\alpha}\ast
v)+(c\ast v))\\
&  =u+D_{T}^{\alpha}(c\ast v)=u,
\end{align*}
because $u$ and the constant function, $c\ast v$, are in the domain of
$D_{T}^{\alpha}$.
\end{proof}

\begin{Theorem}
\label{fund solu Tai} A fundamental solution of (\ref{eq pseudo}) is
\[
E_{\alpha}(x)=%
\begin{cases}
\dfrac{1-p^{-\alpha}}{1-p^{\alpha-n}}||x||_{p}^{\alpha-n} & \text{ if
$\alpha\neq n$}\\
\dfrac{1-p^{n}}{p^{n}\ln p}\ln(||x||_{p}) & \text{ if $\alpha=n$.}%
\end{cases}
\]

\end{Theorem}

\begin{proof}
The proof is based on the ideas introduced in \cite{Z-G1}. The existence of a
fundamental solution $E_{\alpha}$ is equivalent to the existence of a
distribution $\mathcal{F}E_{\alpha}$ satisfying
\begin{equation}
||x||_{p}^{\alpha}\mathcal{F}E_{\alpha}=1, \label{eq in distr}%
\end{equation}
as distributions. Let $||x||_{p}^{s}=\sum\limits_{m\in\mathbb{Z}}%
c_{m}(s+\alpha)^{m}$ be the Laurent expansion at $-\alpha$ with $c_{m}\in
S^{\prime}$ for all $m$. The existence of this expansion is a consequence of
the completeness of $S^{\prime}$ (see e.g. \cite[pp. 65-66]{Igusa}). Since the
real parts of the poles of the meromorphic continuation of $||x||_{p}^{s}$ are
negative rational numbers (cf. Remark \ref{continuation of |x|}),
$||x||_{p}^{s+\alpha}=||x||_{p}^{\alpha}||x||_{p}^{s}$ is holomorphic at
$s=-\alpha$. Therefore, $||x||_{p}^{\alpha}c_{m}=0$ for all $m<0$ and
\[
||x||_{p}^{s+\alpha}=||x||_{p}^{\alpha}c_{0}+\sum_{m=1}^{\infty}%
||x||_{p}^{\alpha}c_{m}(s+\alpha)^{m}.
\]

By using the Lebesgue dominated convergence theorem, one verifies that
\[
\underset{s\rightarrow-\alpha}{\text{lim}}\langle||x||_{p}^{s+\alpha}%
,\phi\rangle=\int_{\mathbb{Q}_{p}^{n}}\phi(x)\,dx=\langle1,\phi\rangle,
\]
and then we can take $\mathcal{F}E_{\alpha}=c_{0}$. Furthermore, if $-\alpha$
is not a pole of $||x||_{p}^{s}$,
\begin{equation}
\mathcal{F}E_{\alpha}=\underset{s\rightarrow-\alpha}{\text{lim}}||x||_{p}^{s}.
\label{F E_a limit}%
\end{equation}

To calculate $c_{0}$, consider the following two cases.

\textbf{Case $\mathbf{\alpha\neq n}$}.

We use (\ref{F E_a limit}) and the Lemma \ref{Fourier trans of Riesz}, i.e.,
\[
\int_{(\mathbb{Q}_{p}^{\times})^{n}} ||x||_{p}^{s} \mathcal{F}(\varphi)(x)\,
dx = \dfrac{1-p^{s}}{1-p^{-s-n}}\int_{(\mathbb{Q}_{p}^{\times})^{n}}
||x||_{p}^{-s-n} \varphi(x)\, dx,
\]
for $s\neq n+(2\pi i/\ln p) \mathbb{Z}$. If $\alpha\neq n$, by
(\ref{F E_a limit}),
\[
\langle E_{\alpha},\mathcal{F}(\varphi) \rangle=\underset{s \rightarrow
-\alpha}{\text{lim}}\int_{(\mathbb{Q}_{p}^{\times})^{n}} ||x||_{p}^{s}
\mathcal{F}(\varphi)(x)\, dx.
\]

If $\alpha>n$, by the Lebesgue dominated convergence theorem, we can
interchange the limit and the integral. If $0<\alpha<n$, by taking into
account that
\[
\int_{||x||_{p}\leq1}||x||_{p}^{\alpha-n}\,dx<+\infty,\quad\text{for }%
0<\alpha<n,
\]
and by using Lebesgue dominated convergence theorem, we can exchange the limit
and the integral. Therefore,
\begin{align*}
\langle E_{\alpha},\varphi\rangle &  =\frac{1-p^{-\alpha}}{1-p^{\alpha-n}}%
\int_{(\mathbb{Q}_{p}^{\times})^{n}}||x||_{p}^{\alpha-n}\varphi(x)\,dx\\
&  =\frac{1-p^{-\alpha}}{1-p^{\alpha-n}}\int_{\mathbb{Q}_{p}^{n}}%
||x||_{p}^{\alpha-n}\varphi(x)\,dx.
\end{align*}
Set $\overset{\sim}{\varphi}(x)=\varphi(-x)$, with $\varphi\in S$. The results
follow by replacing $\varphi$ by $\mathcal{F}(\overset{\sim}{\varphi})$
because $\mathcal{F}\bigl(\mathcal{F}(\overset{\sim}{\varphi})\bigr)=\varphi$.

\textbf{Case $\mathbf{\alpha= n}$}.

We compute the constant term, $c_{0}$, in the expansion
\[
\langle||x||_{p}^{s},\mathcal{F}(\varphi)\rangle=\sum_{m\in\mathbb{Z}}\langle
c_{m},\mathcal{F}(\varphi)\rangle(s+n)^{m}.
\]
Since
\begin{align*}
\langle||x||_{p}^{s},\mathcal{F}(\varphi)\rangle &  =\frac{1-p^{s}}%
{1-p^{-s-n}}\int_{\mathbb{Q}_{p}^{n}}||x||_{p}^{-s-n}\varphi(x)\,dx\\
&  =(1-p^{s})\int_{\mathbb{Q}_{p}^{n}}\frac{p^{v(x)(s+n)}}{1-p^{-s-n}}%
\varphi(x)\,dx,
\end{align*}
where $x=(x_{1},\dotsc,x_{n})$, $v(x):=\underset{1\leq i\leq n}{\min}v(x_{i}%
)$, and $||x||_{p}=p^{-v(x)}$, by expanding
\begin{align*}
\dfrac{(1-p^{s})p^{v(x)(s+n)}}{1-p^{-s-n}}  &  =\dfrac{1-p^{-n}}{\ln
p}(s+n)^{-1}\\
&  +\dfrac{(1-p^{-1})v(x)\ln p-\frac{\ln p}{p}+\frac{(p-1)}{2p}\ln p}{\ln
p}+O(\left(  s+n\right)  ),
\end{align*}
one gets
\[
\langle E_{n},\varphi\rangle=\langle c_{0},\varphi\rangle=\int_{\mathbb{Q}%
_{p}^{n}}\Bigl(\frac{1-p^{n}}{p^{n}\ln(p)}\ln(||x||_{p})+\frac{p^{n}-3}%
{2p^{n}}\Bigr)\varphi(x)\,dx.
\]

The announced results follow by replacing $\varphi$ by $\mathcal{F}%
(\overset{\sim}{\varphi})$, $\varphi\in S$, and using the fact that the
fundamental solution is determined up to the addition of a constant (cf. Lemma
\ref{funda + k}).
\end{proof}

In the case $n=1$, the previous result is already known, see e.g.
\cite[Theorem 2.1]{KO4}.

\section{Fundamental Solutions for Elliptic Operators}

\begin{Theorem}
\label{fund solu ellip} Let $f(D,\alpha)$ be an elliptic operator of order
$d$. Then, a fundamental solution $E_{\alpha}$ of $f(D,\alpha)u=v$, $\alpha
>0$, and $v\in\mathcal{W}$, is given by
\[
E_{\alpha}(x)=%
\begin{cases}
\dfrac{L(p^{\alpha})(1-p^{-d\alpha})}{(1-p^{-n})(1-p^{d\alpha-n})}%
||x||_{p}^{d\alpha-n} & \text{ as a distribution on $\mathcal{W}$, with
$\alpha\neq n/d$}\\
\, & \\
\dfrac{L(p^{n/d})(1-p^{n})}{(1-p^{-n})(p^{n}\ln p)}\ln(||x||_{p}) & \text{ as
a distribution on $\mathcal{W}$, with $\alpha=n/d$},
\end{cases}
\]
where $L(p^{-s})$ is the numerator of $Z(s,f)$.
\end{Theorem}

\begin{proof}
As we mention before, the problem of the existence of a fundamental solution,
$E_{\alpha}$, is equivalent to the existence of a distribution $\mathcal{F}
E_{\alpha}$ satisfying
\[
|f|_{p}^{\alpha}\mathcal{F}E_{\alpha}=1 \text{ in }S^{\prime}.
\]

\noindent By Lemma \ref{elliptic = Taibleson},
\[
\langle|f|_{p}^{\alpha},\varphi\rangle=\langle\frac{(1-p^{d\alpha
})L(p^{-\alpha})}{(1-p^{-n})(1-p^{-d\alpha-n})}\mathit{R}_{d\alpha+n}%
,\varphi\rangle
\]
$\varphi\in\mathcal{W}$, $s\in\mathbb{C}$. The result follows by reasoning as
in the proof of Theorem \ref{fund solu Tai}, and by the fact that the space
$\mathcal{W}$ is invariant under the Fourier transform.
\end{proof}

\begin{Corollary}
\label{cor1}With the hypotheses of the previous theorem, and assuming that
$\alpha\neq n/d$, we have
\[
|\mathcal{F}(E_{\alpha}\ast\varphi)(x)|\leq C(\alpha)||x||_{p}^{-d\alpha
}|\mathcal{F}(\varphi)(x)|,
\]
for all $x\in\mathbb{Q}_{p}^{n}$, and $\varphi\in\mathcal{W}$.
\end{Corollary}

\section{Solutions of Elliptic Pseudo-Differential Equations in Sobolev
Spaces}

Given $\phi\in\mathcal{S}$ and $l$ a non-negative number, we define
\[
||\phi||_{H^{l}}^{2}=\int_{\mathbb{Q}_{p}^{n}}[\text{max}(1,||\xi||_{p}%
)]^{2l}|\mathcal{F}(\phi)(\xi)|^{2}\,d\xi.
\]

We call the completion of $\mathcal{S}$ with respect to $||\cdot||_{H^{l}}$
the \textit{$l$-Sobolev space } $H^{l}:=H^{l}(\mathbb{Q}_{p}^{n})$.

We note that $H^{l}$ contains properly the space of test functions, $S$.
Indeed, consider the function
\[
f(x)=
\begin{cases}
0 & \text{ if }||x||_{p}\leq1\\
||x||_{p}^{-\beta} & \text{ if }||x||_{p}>1
\end{cases}
\]
with $\beta>n$. A direct calculation shows that
\[
||f||_{H^{l}}^{2}=\int\limits_{||\xi||_{p}\leq1}\Bigl|\frac{(1-p^{-n}%
)(1-||\xi||_{p}^{\beta-n}p^{n-\beta})}{(1-p^{n-\beta})}-p^{-\beta}||\xi
||_{p}^{\beta-n}\Bigr|^{2}\,d\xi.
\]
Thus, $||f||_{H^{l}}^{2}<\infty$, but $f$ does not have compact support.

\begin{Lemma}
\label{Sobolev embedding Lemma} If $l>n/2$, then there exists an embedding of
$H^{l}$ into the space of uniformly continuous functions.
\end{Lemma}

\begin{proof}
Let $\phi\in H^{l}$. Since the Fourier transform of a function in $L^{1}$ is
uniformly continuous, it is \ sufficient to show that $\mathcal{F}(\phi)\in
L^{1}$. By using the H\"{o}lder inequality and the fact that%
\[
\int_{\mathbb{Q}_{p}^{n}}(\text{max}(1,||\xi||_{p}))^{-2l}\,d\xi
<+\infty\text{, for }l>n/2\text{,}%
\]
we have%
\[
\int_{\mathbb{Q}_{p}^{n}}|\mathcal{F}(\phi)(\xi)|\,d\xi=\int_{\mathbb{Q}%
_{p}^{n}}\frac{(\text{max}(1,||\xi||_{p}))^{l}}{(\text{max}(1,||\xi
||_{p}))^{l}}|\mathcal{F}(\phi)(\xi)|\,d\xi\leq C||\phi||_{H^{l}}.
\]

\end{proof}

\begin{Lemma}
\label{Lemma 6} For any $\alpha>0$ and $l\geq0$, the mapping $f(D,\alpha
):H^{l+d\alpha}\rightarrow H^{l}$ is a well-defined continuous mapping between
Banach spaces.
\end{Lemma}

\begin{proof}
Let $\phi\in S$. Since $f(D,\alpha)$ is an elliptic operator, by Lemma
\ref{Lemma 1}, we have that
\begin{align*}
||f(D,\alpha)\phi||_{H^{l}}^{2}  &  =\int_{\mathbb{Q}_{p}^{n}}[\text{max}%
(1,||\xi||_{p})]^{2l}|f(\xi)|^{2\alpha}|\mathcal{F}(\phi)(\xi)|^{2}\,d\xi\\
&  \leq C_{1}\int_{\mathbb{Q}_{p}^{n}}[\text{max}(1,||\xi||_{p})]^{2(l+\alpha
)}|\mathcal{F}(\phi)(\xi)|^{2}\,d\xi=C_{1}||\phi||_{H^{l+\alpha}}^{2}.
\end{align*}
The result follows from the fact that $S$ is dense in $H^{l+d\alpha}$.
\end{proof}

\begin{Remark}
\label{remark3}Let $\beta$ be a positive real number, and let
\[
I(\beta):=\int\limits_{||\varepsilon||_{p}\leq1}||\varepsilon||_{p}^{\beta
}\,d\varepsilon.
\]
Then
\[
I(\beta)=\dfrac{1-p^{-n}}{1-p^{-n-n\beta}},\text{ for }\beta>-n.
\]
Indeed,
\begin{align*}
I(\beta) &  =\int\limits_{||\varepsilon||_{p}<1}||\varepsilon||_{p}^{\beta
}\,d\varepsilon+\int\limits_{||\varepsilon||_{p}=1}\,d\varepsilon\\
&  =\int\limits_{||\varepsilon||_{p}<1}||\varepsilon||_{p}^{\beta
}\,d\varepsilon+1-p^{-n}.
\end{align*}
By making the change of variables $\varepsilon_{i}=px_{i}$, $i=1,\dotsc,n$, we
have
\[
I(\beta)=p^{-n-n\beta}I(\beta)+1-p^{-n}.
\]

\end{Remark}

\begin{Theorem}
\label{mainresult}Let $f(D,\alpha)$, $0<\alpha<n/2d$ be an elliptic
pseudo-differential operator of order $d$. Let $l$ be a positive real number
satisfying $l>n/2$. Then, the equation
\[
f(D,\alpha)u=v\quad\left(  v\in S \right)  ,
\]
has a unique uniformly continuous solution $u\in H^{l+d\alpha}$.
\end{Theorem}

\begin{proof}
Let $v\in\mathcal{S}$, then $v=v_{\mathcal{W}}+v_{\mathcal{L}}$, where
$v_{\mathcal{W}}\in\mathcal{W}$ and $v_{\mathcal{L}}\in\mathcal{L}$. Thus, in
order to prove the existence of a solution $u$, it is sufficient to show that
the two following equations have solutions:
\begin{equation}
f(D,\alpha)u_{\mathcal{W}}=v_{\mathcal{W}}, \label{Eq v_W}%
\end{equation}%
\begin{equation}
f(D,\alpha)u_{\mathcal{L}}=v_{\mathcal{L}}. \label{Eq v_L}%
\end{equation}

We first consider equation (\ref{Eq v_W}). By Theorem \ref{fund solu ellip},
$u_{\mathcal{W}}=E_{\alpha}\ast v_{\mathcal{W}}$ is a solution of
(\ref{Eq v_W}), and by Corollary \ref{cor1}, we have%

\begin{align*}
||u_{\mathcal{W}}||_{H^{l+d\alpha}}^{2}  &  =\int_{\mathbb{Q}_{p}^{n}%
}[\text{max}(1,||\xi||_{p})]^{2(l+d\alpha)}|\mathcal{F}(E_{\alpha}\ast
v_{\mathcal{W}})(\xi)|^{2}\,d\xi\\
&  =C(\alpha,d,n)\int_{\mathbb{Q}_{p}^{n}}[\text{max}(1,||\xi||_{p}%
)]^{2(l+d\alpha)}|\left\Vert \xi\right\Vert _{p}^{-2d\alpha}|\mathcal{F}%
(v_{\mathcal{W}})(\xi)|^{2}\,d\xi\\
&  =C(\alpha,d,n)\Bigl\{\int_{||\xi||_{p}\leq1}||\xi||_{p}^{-2d\alpha
}|\mathcal{F}(v_{\mathcal{W}})(\xi)|^{2}\,d\xi\\
&  \hspace{1cm}+\int_{||\xi||_{p}>1}||\xi||_{p}^{2l}|\mathcal{F}%
(v_{\mathcal{W}})(\xi)|^{2}\,d\xi\Bigr\}.
\end{align*}

We now recall that $v_{\mathcal{W}}(\xi)=p^{rn}C\chi_{r}(\xi)$, with $r>0$.
Then, $\mathcal{F}(v_{\mathcal{W}})(\xi)=C\chi_{-r}(\xi)$ and
\begin{align*}
||u_{\mathcal{W}}||_{H^{l+d\alpha}}^{2} &  \leq C(\alpha,d,n)\Bigl\{C^{2}%
p^{2rn}\int\limits_{||\varepsilon||_{p}\leq1}||\varepsilon||_{p}^{-2d\alpha
}\,d\varepsilon\\
&  \hspace{2.5cm}+||v_{\mathcal{W}}||_{H^{l}}^{2}\Bigr\}\\
&  \leq C(\alpha,d,n)\Bigl\{C_{1}(\alpha,d,n)+||v_{\mathcal{W}}||_{H^{l}}%
^{2}\Bigr\},
\end{align*}
since $-2d\alpha>-n$, cf. Remark \ref{remark3}. Therefore $u_{\mathcal{W}}\in
H^{l+d\alpha}$.

We now consider equation (\ref{Eq v_L}). Since
\[
\mathcal{F}(u_{\mathcal{L}})=\mathcal{F}(v_{\mathcal{L}})|f|_{p}^{-\alpha},
\]
and $f$ is elliptic,
\[
|\mathcal{F}(u_{\mathcal{L}})(\xi)|\leq C||\xi||^{-d\alpha}|\mathcal{F}
(v_{\mathcal{L}})(\xi)|, \text{(cf. Lemma \ref{Lemma 1})}.
\]
Then,
\begin{align*}
||u_{\mathcal{L}}||_{H^{l+d\alpha}}^{2}  &  \leq\int_{||\xi||_{p}\leq1}%
||\xi||_{p}^{-2d\alpha}|\mathcal{F}(v_{\mathcal{L}})(\xi)|^{2}\,d\xi\\
&  \hspace{1cm}+\int_{||\xi||_{p}>1}||\xi||_{p}^{2l}|\mathcal{F}%
(v_{\mathcal{L}})(\xi)|^{2}\,d\xi.
\end{align*}

The second integral is bounded by $||v_{\mathcal{L}}||_{H^{l}}^{2}$. For the
first integral, we observe that if $0<\alpha<n/2d$, then
\[
\int_{||\xi||_{p}\leq1}||\xi||_{p}^{-2d\alpha}|\varphi(\xi)|^{2}\,d\xi\leq
C||\varphi||_{L^{2}},
\]
for any $\varphi\in S$. Therefore,
\[
||u_{\mathcal{L}}||_{H^{l+d\alpha}}^{2}\leq C||\mathcal{F}(v_{\mathcal{L}%
})||_{L^{2}}+||v_{\mathcal{L}}||_{H^{l}}^{2}.
\]

In this way, we established the existence of $u\in H^{l+d\alpha}$ which is
uniformly continuous, by Lemma \ref{Sobolev embedding Lemma}, such that
$f(D,\alpha)u=v$, for any $v\in S$. Finally, we show that $u$ is unique.
Indeed, if $f(D,\alpha)u^{\prime}=v$, then
\[
f(D,\alpha)(u-u^{\prime})=0,\qquad\text{i.e., }|f|_{p}^{\alpha}\mathcal{F}%
(u-u^{\prime})=0,
\]
and thus $\mathcal{F}(u-u^{\prime})(\xi)=0$ if $\xi\neq0$, since $f$ is
elliptic. Then $\Psi(x\cdot\xi)(u-u^{\prime})(\xi)=0$ almost everywhere, and a
fortiori $(u-u^{\prime})(\xi)=0$ almost everywhere, and by the continuity of
$u-u^{\prime}$, $u(\xi)=u^{\prime}(\xi)$ for any $\xi\in\mathbb{Q}_{p}^{n}$.
\end{proof}

\section{Solutions of Elliptic Pseudo-Differential Equations in Singular
Sobolev Spaces}

In this section, we modify the Sobolev norm to obtain spaces of functions on
which $f(D,\alpha)$ gives a surjective mapping.

\begin{Definition}
Given $\varphi\in S$ and $l$ a non-negative number, we set
\[
||\varphi||_{\mathcal{H}^{l}}^{2}:=\int_{\mathbb{Q}_{p}^{n}}||\xi||_{p}%
^{2l}|\mathcal{F}(\varphi)(\xi)|^{2}\,d\xi.
\]
We call the completion of $S$ with respect to $||\cdot||_{\mathcal{H}^{l}}$
the \emph{$l$-singular Sobolev Space} $\mathcal{H}^{l}:=\mathcal{H}%
^{l}(\mathbb{Q}_{p}^{n})$. Note that $H^{l}\subseteq\mathcal{H}^{l}$, $l\geq
0$, since $||\varphi||_{\mathcal{H}^{l}}\leq||\varphi||_{H^{l}}$.
\end{Definition}

\begin{Lemma}
\label{lemma3}For any, $\alpha>0$, $l\geq0$, the mapping $f(D,\alpha
):\mathcal{H}^{l+d\alpha}\rightarrow\mathcal{H}^{l}$ is a well-defined
continuous mapping between Banach Spaces.
\end{Lemma}

\begin{proof}
Similar to the proof of Lemma \ref{Lemma 6}.
\end{proof}

We denote by $\mathcal{L}^{l}$ and $\mathcal{W}^{l}$, the respective
completions of $\mathcal{L}$ and $\mathcal{W}$ with respect to $||\cdot
||_{\mathcal{H}^{l}}$; furthermore, we set
\[
\mathcal{H}_{0}^{l}:=\mathcal{L}^{l}+\mathcal{W}^{l}\subseteq\mathcal{H}^{l}.
\]

\begin{Proposition}
\label{prop1}Let $f(D,\alpha)$, $\alpha>0$, be an elliptic pseudo-differential
operator of order $d$, and let $l$ be a non-negative real number. Then
$f(D,\alpha):\mathcal{H}^{l+d\alpha}\rightarrow\mathcal{W}^{l}$, is a
surjective mapping between Banach spaces.
\end{Proposition}

\begin{proof}
By Lemma \ref{lemma3}, the mapping is well-defined. Let $v\in\mathcal{W}^{l}$,
and let $\{v_{n}\}$ a Cauchy sequence in $\mathcal{W}$ converging to $v$. By
Theorem \ref{fund solu ellip}, there exits a sequence $\{u_{n}\}$ in
$H^{l+d\alpha}$ such that $f(D,\alpha)u_{n}=v_{n}$. We now show that
$\{u_{n}\}$ is a Cauchy sequence in $\mathcal{H}^{l+d\alpha}$ as follows:
\begin{align*}
||u_{n}-u_{m}||_{\mathcal{H}^{l+d\alpha}}^{2}  &  \leq C\int_{\mathbb{Q}%
_{p}^{n}}||\xi||_{p}^{2(l+d\alpha)}||\xi||_{p}^{-2d\alpha}|\mathcal{F}%
(v_{n}-v_{m})(\xi)|^{2}\,d\xi\\
&  \leq C||v_{n}-v_{m}||_{\mathcal{H}^{l}}^{2}.
\end{align*}
Thus, there exists $u\in\mathcal{H}^{l+d\alpha}$ such that $u_{n}\rightarrow
u$, and by the continuity of $f(D,\alpha)$, $f(D,\alpha)u=v$.
\end{proof}

\begin{Proposition}
\label{prop2}Let $f(D,\alpha)$, $\alpha>0$, be an elliptic pseudo-differential
operator of order $d$, and let $l$ be a non-negative real number. Then,
$f(D,\alpha):\mathcal{H}^{l+d\alpha}\rightarrow\mathcal{L}^{l}$ is a
surjective mapping between Banach spaces.
\end{Proposition}

\begin{proof}
By Lemma \ref{lemma3}, the mapping is well-defined. Let $v\in\mathcal{L}^{l}$,
and let $\{v_{n}\}$ a Cauchy sequence in $\mathcal{L}$ converging to $v$. By
the same reasoning given in proof Theorem \ref{mainresult} for establishing
the existence of a solution for equation (\ref{Eq v_L}), we obtain a sequence
$\{u_{n}\}$ in $H^{l+d\alpha}$ such that $f(D,\alpha)u_{n}=v_{n}$. We now show
that $\{u_{n}\}$ is a Cauchy sequence in $\mathcal{H}^{l+d\alpha}$.

By using
\[
|\mathcal{F}(u_{n})(\xi)|\leq C||\xi||^{-d\alpha}|\mathcal{F}(v_{n})(\xi)|,
\]
one gets
\begin{align*}
||u_{n}-u_{m}||_{\mathcal{H}^{l+d\alpha}}^{2}  &  \leq C\int_{\mathbb{Q}%
_{p}^{n}}||\xi||_{p}^{2(l+d\alpha)}||\xi||_{p}^{-2d\alpha}|\mathcal{F}%
(v_{n}-v_{m})(\xi)|^{2}\,d\xi\\
&  \leq C||v_{n}-v_{m}||_{\mathcal{H}^{l}}^{2}.
\end{align*}
Thus, there exists $u\in\mathcal{H}^{l+d\alpha}$ such that $u_{n}\rightarrow
u$, and by the continuity of $f(D,\alpha)$, $f(D,\alpha)u=v$.
\end{proof}

From the previous two lemmas we obtain the following result.

\begin{Theorem}
\label{mainII}Let $f(D,\alpha)$ be an elliptic pseudo-differential operator of
order $d$. Let $l$ be a positive \ real number. Then the equation
\[
f(D,\alpha)u=v,\quad v\in\mathcal{H}_{0}^{l}%
\]
has a unique solution $u\in\mathcal{H}^{l+d\alpha}$.
\end{Theorem}

\end{document}